\documentclass[preprints,article,accept,moreauthors,pdftex]{Definitions/mdpi} 

\usepackage{bm}
\usepackage{bbm}
\usepackage{listings}
\usepackage{subfigure}
\usepackage{amsmath}
\usepackage{amssymb}
\usepackage{booktabs}
\usepackage{hyperref}
\usepackage{fourier}
\usepackage[T1]{fontenc}
\usepackage[ruled, lined, linesnumbered, commentsnumbered, longend]{algorithm2e}

	\newtheorem{prop}{Proposition}
		\newtheorem{rem}{Remark}
\hypersetup{colorlinks,citecolor=blue}

\firstpage{1} 
\pubvolume{xx}
\issuenum{1}
\articlenumber{1}
\pubyear{2020}
\copyrightyear{2020}
\history{}

\newcommand*\diff{\mathop{}\!\kern0pt\mathrm{d}}

\Title{Double sweep LU decomposition for American options under negative rates}
\Author{Fabien Le Floc'h}
\AuthorNames{Fabien Le Floc'h}
\address{}
\corres{Correspondence: fabien@2ipi.com}
\abstract{The classic Brennan-Schwartz algorithm to solve the linear complementary problem, which arises from the finite difference discretization of the partial differential equation related to American option pricing does not lead to the exact solution under negative interest rates. This is due to the two exercise boundaries which may appear under negative interest rate, while the algorithm was proven to lead to the exact solution in the case of a single exercise boundary only. This paper explains that two sweeps of the Brennan-Schwartz algorithm in two directions is enough to recover the exact solution. }
\keyword{American options; finite difference method; quantitative finance; pricing.}

\begin{document}
	
\section{Introduction}
American options allow the holder of the contract to exercise their right to buy (for a call option) or sell (for a put option) the underlying asset $S$ at a fixed strike price $K$, at any time prior to the maturity date $T$ of the option contract. In contrast to the European option, where exercise is only possible at the maturity date, the early-exercise feature introduces a non-linearity in the valuation of American options and numerical techniques must be used.

A common technique to price American option contracts is to discretize the partial differential equation (PDE) of the chosen model, such as Black-Scholes \citep{black1973pricing}, local volatility \citep{dupire1994pricing}, or stochastic volatility \citep{heston93}, with the finite difference method \citep{WiDeHo93, OSullivan09, ikonen2007pricing,  lefloch2014tr, le2021pricing}. Then, a linear complementary problem (LCP) must be solved at each time-step. 

In the context of implicit finite difference schemes, there are many ways to solve the LCP: the Brennan-Schwartz algorithm \citep{brennan1977valuation}, front-tracking \citep{pantazopoulos1998front}, the penalty method \citep{nielsen2002penalty}, operator splitting \citep{ikonen2004operator}, the projected SOR \citep{WiDeHo93}, and more recently, the policy iteration of \citet{reisinger2012use}. 
The simplest way is to solve the system without considering the free boundary and then to apply the early exercise condition explicitly through \texttt{currentPrice  = max(payoff, currentPrice)}. While this keeps the second order accuracy on explicit schemes, it is only first order accurate in time on implicit schemes \citep{OSullivan09}. The Brennan-Schwartz algorithm, for its performance and simplicity, is perhaps the most popular algorithm to solve the discrete LCP exactly, in the most common case of a tridiagonal system. But it suffers from known shortcomings \citep{jaillet1990variational}, it does not work for slightly more exotic American contracts, typically with a non-monotonic payoff such as $F(x)= {\lvert x-K\lvert} $, and may also break under negative interest rates for a regular vanilla American option. The policy iteration algorithm resolves those shortcomings. 

In this paper, we propose an alternative, non-iterative algorithm which still works under negative interest rates, as well as for non-monotonic payoffs. The main idea is to apply the Brennan-Schwartz algorithm in two sweeps: one downward sweep (the classic sweep for an American put option) and one upward sweep (the classic sweep for an American call option). Instead of using the original Brennan-Schwartz algorithm, we prefer the LU decomposition formulation of \citet{ikonen2007pricing}, as the LU decomposition stage does not necessarily need to be done every single time, and may thus lead to some interesting performance improvements.

\section{The LCP under the Black-Scholes model}
\subsection{PDE formulation}

Let $\mathcal{L}$ be the Black-Scholes-Merton operator defined by:
\begin{equation}
	\mathcal{L}\left(f(x,t),x,t\right) =  - \frac{1}{2}\sigma(x,t)^2 x^2\frac{\partial^2 f }{\partial x^2}  - \mu(x,t) x \frac{\partial f}{\partial x}+ r(x,t)f(x,t)\,,
\end{equation}
where $x$ is the underlying price, $\mu$ is the underlying drift, $\sigma$ its volatility and $r$ the interest rate, $F(x)=f(x,T)$ the option payoff at maturity and $f(x,t)$ is the option price at time $t$ for an underlying asset spot price of $x$. 

With this notation, the Black-Scholes-Merton equation is
\begin{equation}
	\frac{\partial f}{\partial t}(x,t) = \mathcal{L}\left(f(x,t),x,t\right) \,.
\end{equation}

The early exercise feature of the option adds a free boundary on top of the Black-Scholes-Merton partial differential equation. Let $f$ be the option price, the following system of partial differential inequalities is satified \citep{LaLa96}:

\begin{equation}\label{lcp}
	\left. \begin{gathered}
		\frac{\partial f}{\partial t}(x,t)  \leq 	\mathcal{L}\left(f(x,t),x,t\right)\,,\\
		\left(\frac{\partial f}{\partial t}(x,t) - 	\mathcal{L}\left(f(x,t),x,t\right)\right)\cdot\left(f-F\right)=0\,,\\
		f \geq F\,,
	\end{gathered} \right\}
	\qquad \text{}
\end{equation}
where $(x,t) \in [0,X]\times[0,T]$, with boundary conditions 	
\begin{align}
	f(x,T) &= F(x)\,,\\
	\frac{\partial^2 f}{\partial x^2}(0,t) &= 0\,, \label{eqn:bc-low}\\
	\frac{\partial^2 f}{\partial x^2}(X,t) &= 0\,.\label{eqn:bc-high}
\end{align}

For a vanilla American call, we have $F(x) = \max(x-K,0)$ and for a put, we have $F(x) = \max(K-x,0)$ where $K$ is the strike price.

\subsection{TR-BDF2 Discretization}
For a time discretization defined by $(t_j)_{j \in \{0,..,n\}}  ~~,~~ k_j = t_j-t_{j-1}$ where $t_0=0$ is typically the valuation time and $t_n=T$ the option expiry, the discretization of the Black-Scholes-Merton PDE by the TR-BDF2 scheme reads \citep{lefloch2014tr}
\begin{subequations}
	\begin{align}
	f^\star = f^n + \frac{\alpha k_n}{2}\left(\mathcal{L}(f^n)+\mathcal{L}(f^\star)\right)\,,\\
	f^{n-1} = \frac{1}{2-\alpha}\left(\frac{1}{\alpha} f^\star - \frac{(1-\alpha)^2}{\alpha}f^n + (1-\alpha)k_n \mathcal{L}(f^{n-1})\right)\,,
\end{align}
\end{subequations}
with $\alpha = 2 - \sqrt{2}$ and $f^j(x) = f(x,t_j)$.

A second-order central discretization in space on the full domain, and first-order for the boundary conditions \ref{eqn:bc-low}  and \ref{eqn:bc-high} leads to the following two implicit stages for $j=n,...,1$
\begin{subequations}
	\begin{align}
		M^{j} \bm{f}^{\star} &= \bm{g}^{j}\,,\\
		M^j \bm{f}^{j-1} &= \bm{h}^j\,,
	\end{align}
\end{subequations}
where, at the time-step $j$, $M^j$ is a tridiagonal matrix with lower diagonal $a_{i}^{j}$ for $i \in \{1,...,m\}$, upper diagonal $c_{i}^{j}$ for $i \in \{0,...,m-1\}$ and diagonal $b_{i}^{j}$ for $i \in \{0,...,m\}$ and
\begin{align*}
	a_{i}^{j} & = \frac{\alpha k_{j}}{2\Delta x_{i-1}\left(\Delta x_{i-1}+\Delta x_i\right)}\left(\mu_{j}x_i\Delta x_i  - \sigma_{i,j}^2 x_i^2\right)\,,\\
	b_{i}^{j} & = 1 + \frac{\alpha k_{j}}{2}\left(r_{j}  +\frac{\mu_{j} (\Delta x_{i-1}-\Delta x_i) x_i + \sigma_{i,j}^2 x_i^2}{\Delta x_i \Delta x_{i-1}}\right)\,,\\
	c_{i}^{j} & = -\frac{\alpha k_{j}}{2\Delta x_i\left(\Delta x_{i-1}+\Delta x_i\right)}\left(\mu_{j} x_i \Delta x_{i-1} + \sigma_{i,j}^2 x_i^2\right)\,,\\
	g_{i}^{j}&=-a_{i}^{j} f_{i-1}^{j}+ (2-b_{i}^{j}) f_{i}^{j} - c_{i}^{j} f_{i+1}^{j}\,,\\
	h_{i}^{j}&=\frac{1}{2-\alpha}\left(\frac{1}{\alpha} f_{i}^{\star} - \frac{(1-\alpha)^2}{\alpha} f_{i}^{j}\right)\,,
\end{align*}
for $i=1,...,m-1$ with $\bm{f}^j = (f_{0}^{j},...,f_{m}^{j})^\top$, $\bm{g}^j = (g_{0}^{j},...,g_{m}^{j})^\top$, $\bm{h}^j = (h_{0}^{j},...,h_{m}^{j})^\top$, $\Delta x_i = x_{i+1}-x_i$. The boundary conditions lead to \begin{align*}
	b_{0}^{j} =   1 + \frac{\alpha k_{j}}{2}\left(r_{j}+\frac{\mu_{j} x_0}{\Delta x_0}\right)\,,&\quad	c_{0}^{j}  = - \alpha k_{j} \frac{\mu_{j} x_0}{2 \Delta x_0}\,,\\
	a_{m}^{j} =  \alpha k_{j}\frac{\mu_{j} x_m}{2 \Delta x_{m-1}}\,,&\quad	b_{m}^{j} =  1 + \frac{\alpha k_{j}}{2}\left(r_{j}-\frac{\mu_{j}x_m}{\Delta x_{m-1}}\right)\,.\\
\end{align*}

The  corresponding linear complimentary problem (\ref{lcp}) discretization reads, for $j=n,...,1$
\begin{subequations}\begin{align}\label{lcp_discrete_tr}
	\left.\begin{aligned}
		M^{j} \bm{f}^{\star} &\geq \bm{g}^{j} \\
		\bm{f}^{\star} &\geq F(\bm{x})\\
		\left(M^j \bm{f}^{\star}-\bm{g}^{j}\right)^\top \left(\bm{f}^{\star}-F(\bm{x})\right) &= 0
	\end{aligned}
	\right\}
	&\qquad \text{Trapezoidal stage,}\\
\label{lcp_discrete_bdf2}
	\left.\begin{aligned}
		M^j \bm{f}^{j-1} &\geq \bm{h}^j\\
		\bm{f}^{j-1} &\geq F(\bm{x})\\
		\left(M^j \bm{f}^{j-1}-\bm{h}^j\right)^\top \left(\bm{f}^{j-1}-F(\bm{x})\right) &= 0
	\end{aligned}
	\right\}
	&\qquad \text{BDF2 stage.}
\end{align}
\end{subequations}

The Brennan-Schwartz and the policy iteration algorithms are only valid if the matrix $M^j$ has the following properties \citep{jaillet1990variational}: 
\begin{itemize}
	\item the lower and upper diagonals are negative: $a_{i,j} \leq 0$ and $c_{i,j} \leq 0$ for $i \in \{1,...,m-1\}$, $c_{0,j}\leq0$, $a_{m,j} \leq 0$
	\item the diagonal is dominant: $a_{i,j} + b_{i,j} + c_{i,j} \geq 0$ for  $i \in \{1,...,m-1\}$, $b_{0,j}+c_{0,j} \geq 0$, $a_{m,j}+b_{m,j} \geq 0$ and $b_{i,j}>0$ for $i \in \{0,...,m\}$.
\end{itemize}
In other terms, $M^j$ must be an irreducible Minkowski matrix (also known as M matrix). For our TR-BDF2 discretization, this translates to for $i \in \{1,...,m-1\}$:
\begin{subequations}
\begin{align}\label{eq_mu_bs}
	-\frac{\sigma_{i,j}^2 x_i}{\Delta x_{i-1}} &\leq \mu_{j} \leq \frac{\sigma_{i,j}^2 x_i}{\Delta x_i}\,,\\
\label{eq_r_bs}
	0 &\leq 1+ \frac{\alpha k_j}{2} r_{i,j}\,.
	\end{align}
\end{subequations}
And for the boundaries: 
\begin{align}\label{eq_boundary_bs}
	\mu_{j}x_0 \geq 0 \,,&\quad	\mu_{j}x_m \leq 0\,.
\end{align}
Except for the boundaries, those conditions are almost always verified in practice.
Furthermore one can always make $\Delta x_i$ small enough so that \ref{eq_mu_bs} holds. We may also impose this condition in the general case via exponential fitting \citep{il1969differencing, healy2021}.  If we choose $x_0 = 0$, which also helps in improving the overall accuracy for American options, then only the upper boundary may be problematic.

\section{Double sweep LU decomposition}
The LU decomposition algorithm is slightly easier to analyze on the following reformulated equivalent problem:
\begin{equation}\label{lcp_eo}
	\left.\begin{aligned}
		M^{j} \bm{z} &\geq \bm{v} \\
		\bm{z} &\geq 0\\
		\left(M^j \bm{z}-\bm{v}\right)^\top \bm{z} &= 0
	\end{aligned}
	\right\}
\end{equation}
with 
\begin{align}
	\bm{z} = \bm{f}^\star - F(\bm{x})\,,&\quad \bm{v} = \bm{g}^j - M^j  F(\bm{x})
\end{align}
for the TR-BDF2 stage and
\begin{align}
	\bm{z} = \bm{f}^{j-1} - F(\bm{x})\,,&\quad \bm{v} = \bm{h}^j - M^j  F(\bm{x})
\end{align}
for the BDF2 stage.

The algorithm presented in \cite{ikonen2007pricing}, valid for an American call payoff, will decompose $M^j$ such that $M^j = LU$ with $L$ lower triangular, $U$ upper triangular and solve first $L \bm{y} = \bm{v}$, then $U \bm{z} = \bm{y}$. Similarly, the algorithm of an American put payoff will decompose $M^j$ such that $M^j = \bar{U}\bar{L}$ with $\bar{L}$ lower triangular, $\bar{U}$ upper triangular and solve first $\bar{U} \bm{y} = \bm{v}$, then $\bar{L} \bm{z} = \bm{y}$. It is during the last step that the Brennan-Schwartz algorithm differs from the classic algorithm for a linear tridiagonal system, by enforcing the non-linear constraint.
\begin{algorithm}
	$l_{00} = b_{0}$  \tcp*[f]{start $LU$ decomposition.}\\
	$u_{01} = c_{0}/l_{00}$\\	
	\For{$i \leftarrow $ 1 \KwTo $m-1$}{
			$l_{ii-1} = a_{i}$\\
			$l_{ii} = b_{i} - l_{ii-1}u_{i-1i}$\\
			$u_{ii+1} = c_{i}/l_{ii}$
%
	}
	$l_{mm-1} = a_{m}$\\
	$l_{mm} = b_{m} - l_{mm-1}u_{m-1m}$\\
	$\bar{u}_{mm} = b_{m}$ \tcp*[f]{start $\bar{U}\bar{L}$ decomposition.}\\
	$\bar{l}_{mm-1} = a_{m}/\bar{u}_{mm}$\\	
	\For{$i \leftarrow m-1$ 1 \KwTo $1$}{
		$\bar{u}_{ii+1} = c_{i}$\\
		$\bar{u}_{ii} = b_{i} - \bar{u}_{ii+1}\bar{l}_{i+1i}$\\
		$\bar{l}_{ii-1} = a_{i}/\bar{u}_{ii}$
			%
	}
	$\bar{u}_{01} = c_{0}$\\
	$\bar{u}_{00} = b_{0} - \bar{u}_{01}\bar{l}_{10}$\\
	\caption{LUUL Decomposition for the transformed problem}
\end{algorithm}

\begin{algorithm}
	$y_0 = v_0/l_{00}$     \tcp*[f]{start $LU$ back-solve.}\\
	\For{$i \leftarrow $ 1 \KwTo $m$}{
		$y_i = (v_i - l_{ii-1}y_{i-1})/l_{ii}$
	}

	$z_m = y_m$\\
	$z_m = \max(z_m, 0)$\\
	\For{$i \leftarrow  m-1$ \KwTo $0$}{
		$z_i = y_i - u_{ii+1}z_{i+1}$\\
			$z_i = \max(z_i, 0)$
	}
	$y_m = v_{m}/\bar{u}_{mm}$     \tcp*[f]{start  $\bar{U}\bar{L}$ back-solve.}\\
	\For{$i \leftarrow m-1$  \KwTo 0}{
		$y_i = (v_i - \bar{u}_{ii+1}y_{i+1})/\bar{u}_{ii}$
	}
	
	$\bar{z}_0 = y_0$\\
	$z_0 = \max(z_0, \bar{z}_0)$\\
	\For{$i \leftarrow  1$ \KwTo $m$}{
		$\bar{z}_i = y_i - \bar{l}_{ii-1}z_{i-1}$\\
		$z_i = \max(z_i, \bar{z}_i)$
	}
	\caption{Brennan and Schwartz algorithm with LUUL Decomposition for the transformed problem\label{alg:BS_LUUL}}
\end{algorithm}

The correctness of the original Brennan-Schwartz algorithm is proven in \citep[Propostion 5.6]{jaillet1990variational} when $M^j$ is an M-matrix using specific assumptions on the shape of the solution. For an American put option, on the transformed problem formulation, the assumption reads \begin{equation}
\exists k \in \{0,...,m\} \mid \forall i \leq k, z_i = 0 \textmd{ and } \forall i > k, z_i > 0\,.
\end{equation}
Under positive interest rates, a single early-exercise boundary exists for an American call or put option, and this justifies the validity of the $\bar{U}\bar{L}$ back-solving in Algorithm \ref{alg:BS_LUUL}. In particular, the original Brennan-Schwartz algorithm is not valid in the context of negative interest rates, as for a vanilla American call or put option, two early-exercise boundaries appear when respectively $r < r-\mu < 0$ or  $r-\mu < r < 0$ for a constant interest rate $r$ and drift $\mu$ \citep{andersen2021fast}.

\citet{andersen2021fast} show that the two boundaries may be solved independently. It can also easily be seen from the PDE formulation as a free-boundary problem \citep{healy2021}:
\begin{equation}
	\mathcal{L} f(x,t) = \frac{\partial f}{\partial t}(x,t)\quad \textmd{ for } (x,t) \in 
	\mathcal{C}\,,\label{eqn:am_pde}
\end{equation}
with initial condition 
\begin{align}
	\lim\limits_{t \to T} f(x, t) &= F(x)\,, \label{eqn:am_payoff}
\end{align}
and boundary conditions
\begin{align}
	f(x,t) &= F(x) \quad \textmd{ for } x=u(t) \textmd{ or } x = l(t)\,,\label{eqn:am_ulb}\\
	\frac{\partial f}{\partial x} &= -1\quad \textmd{ for } x=u(t)  \textmd{ or } x=l(t)\,, \label{eqn:am_ulb_high}\\
	f(x,t) &> F(x) \quad \textmd{ in } \mathcal{C}\,,\\
	f(x,t) &= F(x) \quad \textmd{ in } \mathcal{D}\,,\label{eqn:am_d}
\end{align}
where $l(t)$ and $u(t)$ represent respectively the lower and upper early-exercise boundaries and
\begin{align}
	\mathcal{C} &= \left\{ (x,t) \in [0,\infty)\times[0,T] : l(t) < x < u(t) \right\} \,,\\
	\mathcal{D} &= \left\{ (x,t) \in [0,\infty)\times[0,T] : x < l(t) \right\} \cup \left\{ (x,t) \in [0,\infty)\times[0,T] : x > u(t) \right\} \,.
\end{align}

As long as the boundaries do not yet intersect, we can split the problem in two separate domains $\left\{ (x,t) \in [0,\infty)\times[0,T] : x \leq l(t) \right\}$ and $\left\{ (x,t) \in [0,\infty)\times[0,T] : x \geq u(t) \right\}$. When the boundaries intersect at time $t^\dag$, early-exercise is never optimal for $t < t^\dag$, and the algorithm is applicable on $(t^\dag,T]$.

This split motivates the two back-solving sweeps of Algorithm \ref{alg:BS_LUUL}. Now let us prove the validity of the algorithm for the case of two boundaries.
\begin{prop}
	If the solution $\bm{z}$ of System \ref{lcp_eo} satisfies  
	\begin{equation}\exists (k_1,k_2) \in \{0,...,m\}^2 \mid \forall  k_1 \leq i \leq k_2 , z_i = 0 \textmd{ and } \forall i \in \{0,...,m\}\setminus \{k_1,...,k_2\}\,,  z_i > 0\,,\end{equation} and $M^j$ is an M-matrix,
	then Algorithm \ref{alg:BS_LUUL} finds the exact solution.
\end{prop}
\begin{proof}
The $LU$ sweep leads to $z_i$, $i > k_2$ and the $\bar{U}\bar{L}$ sweep to  $z_i$, $i < k_1$. In between we know that $z_i = 0$.
Because $M^j$ is an M-matrix, we know from \citet{cottle1976solution}, that the solution is unique. 
\end{proof}
\begin{rem}
In Algorithm \ref{alg:BS_LUUL}, in general, we can not stop the loops at $k'_1$ and $k'_2$ where $k'_1$ is the first index such that $y_i \leq 0$ and $k'_2$ the first index such that  $\bar{y}_i \leq 0$ in the spirit of the algorithm of \citet[p. 115-116]{elliot1985weak}.
\end{rem}
For example, stopping early breaks when the vector $v$ is such that
\begin{align}
	\begin{cases}
		v_i > 0 \textmd{ for } i < k''_1\,, \\
		v_i < 0 \textmd{ for }  k''_1 \leq i < k''_2\,, \\
		v_i > 0 \textmd{ for }   k''_2 \leq i < k''_3\,, \\
		v_i = 0 \textmd{ for } k''_3 \leq i\,,
		\end{cases}
\end{align}
 for some $k''_1, k''_2, k''_3$ such that $0<k''_1<k''_2<k''_3<m$ as in the case of an American put option where $r-\mu < r < 0$.

The double sweep algorithm will also lead to a very good estimate of the solution for a butterfly American option, while an \citet[p. 115-116]{elliot1985weak} like algorithm will not. The latter would require more transitions. 

\begin{prop}
	For an American call option under positive interest rates, the Brennan-Schwartz algorithm is still valid when $c_0 > 0$ or $a_m > 0$, if $v_0\leq0$.
\end{prop}
\begin{proof}
If $c_0 > 0$, we  have $l_{ii}>0$ as long as $l_{11} = b_1 - a_1 u_{01} = b_1 - a_1 c_0/b_0 > 0$, which is true since $a_1 < 0$ and $b_i > 0$. Thus $y_0>0$ and the sign of $y_i$ is unchanged compared to the case $c_0 < 0$, for $i\geq 1$. Similarly, the sign of $z_i$ is unchanged for $i\geq 1$. Only at $i=0$ we may have an inconsistency, but for a call, it is never optimal to exercise using practical grid bounds as the early exercise payoff is essentially 0.

If $a_m > 0$, we have $l_{mm} = b_m - a_m u_{m-1m} = b_m - a_m* c_{m-1}/l_{m-1m-1}$. We know that $c_{m-1} < 0$ and thus $l_{mm} > 0$, $u_{m-1m} < 0$.  The value of $y_m$ may still be strictly negative. Let  $k_1$  be the first index such that $y_i \leq 0$. In practice, $k_1 < m$, unless it is never optimal to early-exercise. Since $y_i$ becomes negative for $k_1 \leq i < m$, the back-solving loops may be stopped at index $k < m$ and $z_i = 0$. The value $y_m$ is effectively not used.
\end{proof}
The same reasoning is obviously applicable to the Brennan-Schwartz algorithm for the American put, as well as to Algorithm \ref{alg:BS_LUUL}.



In order to improve the performance of the algorithm on trivial cases, we may check for the number of sign changes and the sign of the first element of $\bm{v}$. If there is zero or one sign change, we may run only the $UL$ back-solve when the sign is positive, and  the $LU$ back-solve when the sign is negative.

\section{Numerical examples}
\subsection{Negative interest rates}
We consider the example from \cite{andersen2021fast} of an American put option of strike $K=100$ and various maturities with an underlying asset spot price of $S=100$, a constant interest rate of $r=-1.2\%$, a drift $\mu = 0.4\%$ and volatility $\sigma = 10\%$.

We price each option with the TR-BDF2 scheme applied on a grid composed of $n=100$ steps in the time dimension and $m=2000$ steps in the asset dimension. In the asset dimension, we use a non-uniform hyperbolic grid, with more points close to the strike price and less points at the boundaries. In the time dimension, we consider two different discretizations, one with constant steps, and one with the time step size following a uniform square root law: $t_j = T - (n-j)^2/n^2 T$. In the former time-discretization, the system matrix may be computed once, and the LU factorization may be reused across time-steps latter, while in the latter the system matrix must be updated at each time-step. The latter is representative of the more general case of non-constant rates, drift or volatility. The TR-BDF2 scheme however involves the same matrix decomposition in each of its two internal stages, and the LU factorization is still beneficial in practice.

As expected, the double sweep algorithm takes twice the time of the original Brennan-Schwartz algorithm (Table \ref{tbl:andersen}). 
The tridiagonal policy iteration solver of \citet{reisinger2012use} is around 50\% slower on this example.
\begin{table}[h]
	\caption{Error in the price of an American put of five distinct maturities, when computed with the TR-BDF2 scheme and various solvers for the LCP. PI, LUUL, BS stand respectively for the tridiagonal policy iteration solver, the double sweep LU decomposition solver and the classic Brennan-Schwartz solver.\label{tbl:andersen}}
	\centering{
		\begin{tabular}{lllrrr}\toprule
			$T$ & Reference Price & Time-steps& PI Error (Time) & LUUL Error (Time) & BS Error (Time) \\\midrule
			45/365 & 1.380533089 & Varying &-1.0e-5 (27 ms)& -1.0e-5 (20 ms)& -2.0e-3 (10 ms)\\
& & Constant & -2.9e-5 (29 ms)& -2.9e-5 (14 ms) & -2.0e-3 (9 ms) \\ 
	90/365 &	1.942381237 & Varying & -3.1e-5 (29 ms) & -3.1e-5 (13 ms) & -7.1e-3 (9 ms)\\
	& & Constant &1.0e-7 (30 ms) & 1.0e-7 (19 ms)& -3.8e-3 (10 ms)\\
	180/365 & 2.729267252 & Varying & -8.2e-6 (35 ms)& -8.2e-6 (19 ms) & -7.1e-3 (11 ms)\\
	&& Constant & -5.9e-5 (36 ms)& -5.9e-5 (13 ms) & -7.1e-3 (9 ms)\\
	360/365 & 3.830520425 & Varying & 1.4e-6 (40 ms)& 1.4e-6 (18 ms) & -1.2e-2 (11 ms)\\
	& & Constant & 8.1e-5 (36 ms)& 8.1e-5 (14 ms) & -1.2e-2 (9 ms)\\
	3600/365 &12.189323541 & Varying &-3.6e-6 (29 ms)& -3.6e-6 (17 ms) & -1.4e-2 (10 ms) \\
	&& Constant &-4.5e-4 (25 ms)& -4.5e-4 (10 ms) & -1.4e-2 (9 ms) \\\bottomrule
	\end{tabular}}
\end{table}
We found the policy iteration solver to be even slower (around twice) for call options under negative rates.

Table \ref{tbl:accuracy} verifies, using a small grid, that the double sweep LU decomposition leads to exactly the same solution as the policy iteration solver.

\begin{table}[h]
	\caption{Price of an American option on a small grid of 20 time-steps and 20 space steps, when computed with the TR-BDF2 scheme and various solvers for the LCP. PI, LUUL, BS stand respectively for the tridiagonal policy iteration solver, the double sweep LU decomposition solver and the classic Brennan-Schwartz solver. $S = 90, K = 100, \sigma = 8\%, r = 1\%, \mu=0.5\%$.\label{tbl:accuracy}}
	\centering{
		\begin{tabular}{lrrr}\toprule
			 & PI  & LUUL  & Difference \\\midrule
			Call & 0.2924244529450148 & 0.29242445294501457 & 2.2 e-16\\ 
			Put &10.635776477887287 & 10.635776477887285 & 1.8e-15		\\
\bottomrule
	\end{tabular}}
\end{table}

\subsection{American Butterfly}
We consider now an American butterfly option of strikes $K_1 = 90$ and $K_2 = 110$, maturity $T=0.25$ using the following market data: $r=\mu=1\%, \sigma=100\%, S = 110$. We use a fixed uniform discretization in the asset price dimension composed of 301 points from $x_0=0$ to $x_m = 300$ and vary the number of time-step using uniform steps. In particular, $x_m$ is less than three standard deviations away from the spot price, which allows to put in evidence the error of the double sweep algorithm. The number of points does not change the scale of the error.

Table \ref{tbl:butterfly} shows that the double sweep algorithm is still very accurate in practice. It does not perturb the order of convergence in contrast to the classic Brennan-Schwartz algorithm which results in a significantly larger error in price.
\begin{table}[h]
	\caption{The column "Difference" is the price obtained by the Successive Over Relaxation method subtracted to the price obtained by the specific solver.\label{tbl:butterfly}}
	\centering{
		\begin{tabular}{llrrr}\toprule

		$n$	& Solver & Price & Difference & Time\\\midrule
	4 & BS & 6.163251 & -2.74e+00 & 110µs\\
	& LUUL &  8.900522 & -1.52e-06 & 126 µs\\
	 & PI &  8.900523 & 4.44e-14 &210 µs\\
	8 & BS & 7.030902  & -1.83e+00 & 235 µs\\
	& LUUL & 8.865021 & -2.81e-07 & 228 µs\\
	 &PI & 8.865021 &  -6.04e-14 & 296 µs\\
	 16 &BS & 7.596790 & -1.27e+00 & 306µs\\
&LUUL & 8.863211 & -1.51e-08 & 465 µs \\
	    & PI & 8.863211 & 1.03e-13 & 486 µs\\
	 32  &BS & 7.972106 & -8.91e-01 & 795 µs \\
	 & LUUL & 8.862836 &-1.56e-10 & 964 µs\\
	    & PI & 8.862836 & -1.78e-14 & 806 µs\\
	 64 & BS & 8.248415 & -6.14e-01 & 924 µs\\
	 &LUUL & 8.862750 & -1.79e-13 & 1139 µs\\
	    &PI & 8.862750  & 3.55e-14 & 1387 µs\\
	\bottomrule
	\end{tabular}}
\end{table}
The policy iteration algorithm is however nearly as fast as the double sweep on this example, especially when the number of time-steps is larger than 16. 

\section{Conclusion}
We have shown that using two sweeps of the traditional Brennan-Schwartz algorithm  constitute a simple and exact algorithm to solve the linear complementary problem arising in the pricing of American options under negative interest rates when the system involves a tridiagonal matrix (the most common case in practice). It is particularly relevant to price vanilla American options under non-constant interest rate, or underlying asset drift, as well as for stocks paying discrete dividends.

It is faster in general than the policy iteration algorithm optimized for tridiagonal systems, while being straightforward to implement. It is however not exact anymore for non-monotonic payoffs, such as American butterfly options, but we found it to lead to very accurate results in practice nonetheless.

\externalbibliography{yes}
\bibliography{double_sweep_american.bib}
\appendixtitles{no}
\appendix

\section{Inside the butterfly example}
The inaccuracy of the double sweep algorithm on the American butterfly example is already visible if we reduce the number of steps in the asset price dimension to 15, and use 3 time-steps.

The tridiagonal matrix is as follows:
\begin{align*}
\bm{a} &= (0, -0.012081845276054914, -0.0485714587865642, -0.10946884053152789,\\& -0.19477399051094593, -0.3044869087248183, -0.4386075951731451, -0.5971360498559263,\\& -0.7800722727731618, -0.9874162639248517, -1.219168023310996, -1.4753275509315948,\\& -1.7558948467866478, -2.0608699108761552, -2.3902527432001173, 0.0036611652351682144)\,,\\
\bm{b} &= (1.0002440776823445, 1.024651845916799, 1.097875150620162, 1.219913991792434,\\& 1.3907683694336146, 1.6104382835437039, 1.878923734122702, 2.196224721170609,\\& 2.562341244687425, 2.977273304673149, 3.441020901127782, 3.953584034051324,\\& 4.514962703443775, 5.125156909305134, 5.784166651635402, 0.9965829124471763)\,,\\
\bm{c} &= (0, -0.012325922958399462, -0.0490596141512533, -0.1102010735785615,\\& -0.19575030124032408, -0.30570729713654105, -0.44007206126721243, -0.5988445936323381, \\&-0.7820248942319182, -0.9896129630659527, -1.2216088001344414,\\& -1.4780124054373847, -1.7588237789747825, -2.064042920746634, -2.3936698307529403, 0)\,,		
\end{align*}
with initial vector
\begin{align*}
	\bm{g} &=  (0, 0, 0, 0,\\& 1.9575030124032409, 3.895617164562961, 4.386075951731451, 0,\\& 0, 0, 0, 0,\\& 0, 0 ,0, 0)\,,
\end{align*}
and lower boundary
\begin{align*}
F &= (0, 0, 0, 0, 0, 10, 0, 0, 0, 0, 0, 0, 0, 0, 0, 0)\,.
\end{align*}
The nearly exact solution found by the policy iteration algorithm reads
\begin{align*}
	\bm{f}^\star &= (0, 0.00013908409255599, 0.011562036583884633, 0.2586020570733455,\\& 2.8512116514642054, 10, 5.021713994073349, 1.507175220503007,\\&  0.5200832132225875,0.20066505470872006, 0.0847766655914132, 0.038534316489836615,\\&  0.018454045388137823, 0.008902039586522732, 0.0036786845579122227, 0)\,.
\end{align*}

The error of the double sweep algorithm reads
\begin{align*}
	\bm{f}^\star_{LUUL}-\bm{f}^\star	&= (0,0,	0,	0,	0,	0,	2.53\cdot10^{-9},	1.08\cdot10^{-8},\\&	3.71\cdot10^{-8},	1.11\cdot10^{-7},	2.96\cdot10^{-7},	7.24\cdot10^{-7},	1.64\cdot10^{-6},	3.49\cdot10^{-6},	7.02\cdot10^{-6},0)	\,.
\end{align*}
\section{Combined double-sweep Brennan-Schwartz}
Here we give the more compact algorithm corresponding the double sweep Brennan-Schwartz technique, where the LU decomposition and back-solve are merged. It is not faster
when using the TR-BDF2 scheme as the LU decomposition is reused among the two stages of the scheme, but may be faster for the implicit Euler (BDF1 or BDF2) schemes when the problem to solve includes time-dependent coefficients.

\begin{algorithm}
	$y_0 = b_{0}$     \tcp*[f]{start  fast $LU$ back-solve.}\\
	$z_0 = v_0$\\
	\For{$i \leftarrow 1$  \KwTo $m$}{
		$y_i = b_i - a_{i}c_{i-1}/{y}_{i-1}$\\
		${z}_i = {v}_i - a_i {z}_{i-1}/y_{i-1}$
	}
	
	${z}_m = {z}_m/y_m$\\
	$z_m = \max(z_m, 0)$\\
	\For{$i \leftarrow  m-1$ \KwTo $0$}{
		${z}_i = ({z}_i - c_{i}z_{i+1})/y_i$\\
		$z_i = \max(z_i, 0)$
	}
	
	$y_m = b_{m}$     \tcp*[f]{start  fast $\bar{U}\bar{L}$ back-solve.}\\
	$\bar{z}_m = v_m$\\
	\For{$i \leftarrow m-1$  \KwTo 0}{
		$y_i = b_i - c_{i}a_{i+1}/{y}_{i+1}$\\
		$\bar{z}_i = v_i - c_i \bar{z}_{i+1}/y_{i+1}$
	}
	
	$\bar{z}_0 = \bar{z}_0/y_0$\\
	$z_0 = \max(z_0, \bar{z}_0)$\\
	\For{$i \leftarrow  1$ \KwTo $m$}{
		$\bar{z}_i = (\bar{z}_i - a_{i}z_{i-1})/y_i$\\
		$z_i = \max(z_i, \bar{z}_i)$
	}
	\caption{Fast double sweep Brennan and Schwartz algorithm for the transformed problem\label{alg:BS_LUUL_F}}
\end{algorithm}
\end{document}